\documentclass{article}
\pdfoutput=1
\usepackage{paralist}
\usepackage[textsize=scriptsize]{todonotes}
\usepackage[utf8]{inputenc}
\usepackage{microtype}
\usepackage{amsmath}
\usepackage{amssymb}
\usepackage{amsthm}
\usepackage{mathtools}
\usepackage[colorlinks,linkcolor=red!40!black, citecolor=red!40!black,
bookmarks=false]{hyperref}
\usepackage{authblk}
\usepackage[numbers,sort]{natbib}

\newcommand{\pt}{pol\-y\-no\-mi\-al-time}
\newcommand{\cf}{con\-stant-fac\-tor}
\newcommand{\CF}{Con\-stant-fac\-tor}
\newtheorem{theorem}{Theorem}
\newtheorem{corollary}{Corollary}
\newtheorem{lemma}{Lemma}
\newtheorem{observation}{Observation}
\theoremstyle{definition}
\newtheorem{definition}{Definition}
\newcommand{\CARPlong}{\textsc{Capacitated Arc Routing}}
\newcommand{\CARP}{\textsc{CARP}}
\newcommand{\TSPlong}{\textsc{Traveling Salesperson}}
\newcommand{\TSP}{\textsc{TSP}}

\DeclareMathOperator\dist{dist}
\newcommand{\OPT}{\text{OPT}}
\newcommand{\OPTfulltri}{\text{OPT}^{\blacktriangledown}}

\newcommand{\chalftri}{c^{\triangledown}}
\newcommand{\cfulltri}{c^{\blacktriangledown}}

\newcommand{\prob}[5]{%
  \begin{center}
    \begin{minipage}{\columnwidth}
      #1
      \begin{compactdesc}
      \item[#2]#3
      \item[#4]#5
      \end{compactdesc}
    \end{minipage}
  \end{center}
}

\newcommand{\optprob}[3]{\prob{#1}{Instance:}{#2}{Task:}{#3}}

\makeatletter
\def\NAT@spacechar{~}
\makeatother

  \title{\CF{} approximations for Capacitated Arc Routing
    without triangle inequality}

\date{}

  \author{René van Bevern}

  \author{Sepp Hartung}

  \author{André Nichterlein}

  \author{Manuel~Sorge}

  \affil{Institut f\"ur Softwaretechnik und Theoretische
    Informatik, TU Berlin, Germany, \\\texttt{\{rene.vanbevern, sepp.hartung, andre.nichterlein, manuel.sorge\}@tu-berlin.de}}

\begin{document}
\maketitle
\begin{abstract}
  \noindent Given an undirected graph with edge costs and edge
  demands, the \CARPlong{} problem (\CARP{}) asks for minimum-cost
  routes for equal-capacity vehicles so as to satisfy all demands.
  \CF{} \pt{} approximation algorithms were proposed for \CARP{} with
  triangle inequality, while \CARP{} was claimed to be NP-hard
   to approximate within any constant factor in general.  Correcting
   this claim, we show that any factor~$\alpha$ approximation for
   \CARP{} with triangle inequality yields a factor~$\alpha$
   approximation for the general \CARP{}.
\end{abstract}

\section{Introduction}
\noindent \citet{GolW81} introduced the \CARPlong{} problem in order
to model the search for minimum-cost routes for vehicles of equal
capacity that satisfy all ``customer'' demands.  Herein, ``customers''
are often the roads of a road network and, hence, are modeled as edges
of a graph with corresponding integer demands.  The vertices of the
graph can be thought of as road intersections.

\optprob{\CARPlong{} \textsc{Problem} (\CARP{})}{An undirected graph~$G=(V,E)$, a
  vehicle depot vertex~$v_0\in V$, edge costs~$c(e) \geq 0$ and edge demands~$d(e)\geq 0$ for every~$e \in E$, and a vehicle
  capacity~$W$.}{Find a set~$\mathcal C$ of cycles in~$G$, each
  corresponding to the route of one vehicle and each passing through
  the depot vertex~$v_0$, and a serving function~$s\colon \mathcal
  C\to 2^E$ such that
  \begin{compactenum}
  \item $\sum_{C\in\mathcal C}\sum_{e\in C}c(e)$ is minimized,
  \item each cycle~$C\in\mathcal C$ \emph{serves} a subset~$s(C)$ of
    edges of~$C$ such that
      $\sum_{e\in s(C)}d(e)\leq W$, and
    \item each edge~$e$ with $d(e)>0$ is served by exactly one cycle
      in~$\mathcal C$.
  \end{compactenum}
}
Well-known special cases of \CARP{} are the NP-hard \textsc{Rural
  Postman Problem}~\citep{LenK76}, where the vehicle capacity is
unbounded, and the \pt{} solvable \textsc{Chinese Postman
  Problem}~\citep{Edm65,EdmJ73}, where the vehicle capacity is
unbounded and all edges have positive demand.

\citet{Jan93} and \citet{Woe08c} gave \pt{} factor
$(7/2-3/W)$ approximation algorithms for \CARP{} when the edge cost function
satisfies the triangle inequality. That is, for any two
edges~$\{u,v\}$ and~$\{v,w\}$ there is an edge~$\{u,w\}$ such that
\[
c(\{u,w\})\leq c(\{u,v\})+c(\{v,w\}).
\]
\citet{GolW81} and \citet{Woe08c} claimed that \CARP{} is NP-hard to
approximate within any constant factor~$\alpha>0$.  However, a recent
arc routing survey~\citep{BNSW13} pointed out that the argument
leading to this claim is erroneous, thus calling for an alternative
proof for the inapproximability of \CARP{} or for a \cf{}
approximation.  We find the latter by proving the following theorem:

\newcommand{\mainthm}{%
\begin{theorem}\label{thm:mainthm}
  \CARP{} is \pt{} self-reducible, mapping any instance~$I$ to an instance~$I'$ in such a way that
  \begin{enumerate}[i)]
  \item $I'$~satisfies the triangle inequality and
  \item  a
  factor-$\alpha$ approximate solution for~$I'$ is \pt{}
  transformable into a factor-$\alpha$ approximate solution for~$I$.
  \end{enumerate}
\end{theorem}}
\mainthm{}
\noindent In terms of approximation-preserving \pt{} reductions~\citep{WilS11}, we
more specifically show a $(1,1)$ L-reduction from \CARP{} without
triangle inequality to \CARP{} with triangle inequality{}.

\autoref{thm:mainthm} and the factor $(7/2-3/W)$ approximation
given by \citet{Jan93} and \citet{Woe08c} for \CARP{} with triangle
inequality then immediately yields the following corollary, which
answers Challenge~6 of the above-mentioned arc routing survey~\citep{BNSW13}.

\begin{corollary}
  There is a \pt{} factor $(7/2-3/W)$ approximation for
  \CARPlong{}, even if the edge cost function does not respect the
  triangle inequality.
\end{corollary}
\noindent Before proving \autoref{thm:mainthm}, we quickly recall the erroneous
argument~\citep{GolW81,Woe08c} for the approximation hardness of
\CARP{} without triangle inequality.

\section{Erroneous argument towards approximation hardness}
\noindent \citet{GolW81} and \citet{Woe08c} claim that \CARP{} without
triangle inequality is NP-hard to approximate within any constant
factor~$\alpha>0$.  Their claim is based on the fact that the
\TSPlong{} problem (\TSP{}) without triangle inequality is NP-hard to
approximate within any constant factor~$\alpha>0$~\citep{SahG76}.

\looseness=-1 They use the following \pt{} transformation from \TSP{} to
\CARP{}. Given a \TSP{} instance, split each vertex of the input
\TSP{} graph and join them by an edge of demand one and cost zero and
set the vehicle capacity~$W$ to be at least the number of input
vertices.  The remaining edges in the \CARP{} instance inherit their
cost from the input \TSP{} instance and have demand zero.

Clearly, any solution for the input \TSP{} instance translates into a
solution for \CARP{} with the same cost.  However, the reverse is not
true: while \TSP{} allows every vertex to be visited at most once,
\CARP{} imposes no such restrictions. Hence, in order to reach a
positive-demand edge from another positive-demand edge, always a
shortest path may be used in the optimal \CARP{}~tour.
A~counterexample to the correctness of the above reduction is given in \autoref{fig:counterex}.
\begin{figure}[t]
  \centering
  \begin{tikzpicture}[node distance=2.6cm and 2.6cm]
    \tikzstyle{vert}=[circle,draw,fill=black,minimum size=3pt,inner sep=0pt]
    \tikzstyle{touredge}=[]
    \tikzstyle{nontouredge}=[dotted]
    \tikzstyle{reqedge}=[very thick]

    \node[vert] (v1) {};
    \node[vert] (v2) [below=of v1] {}
    edge[touredge] node[left] {$\ell$} (v1);
    \node[vert] (v3) [right=of v2] {}
    edge[nontouredge]node[left,pos=0.7] {$\ell$} (v1)
    edge[touredge] node[above] {1} (v2);
    \node[vert] (v4) [above=of v3] {}
    edge[touredge] node[below] {1} (v1)
    edge[nontouredge]node[left, pos=.7] {$\ell$} (v2)
    edge[touredge] node[right] {1} (v3);
    
    \begin{scope}[node distance=2cm and 2cm]
      \node[vert, yshift=-.3cm] (v6) [right=2cm of v4] {}; 
      \node[vert] (v5) [above right=.4cm of v6] {}
      edge[reqedge] (v6);
      \node[vert] (v7) [below=of v6] {}; 
      \node[vert] (v8) [below right=.4cm of v7] {}
      edge[reqedge] (v7);
      \node[vert] (v9) [right=of v8] {}
      edge[touredge] (v7)
      edge[touredge] (v8);
      \node[vert] (v10) [above right=.4cm of v9] {}
      edge[reqedge] (v9);
      \node[vert] (v11) [above=of v10] {}
      edge[touredge] (v10);
      \node[vert] (v12) [above left=.4cm of v11] {}
      edge[reqedge] (v11)
      edge[touredge] (v9)
      edge[touredge] (v6)
      edge[touredge] (v5);

      \draw[nontouredge] (v5)--(v7);
      \draw[nontouredge] (v6)--(v7);
      \draw[nontouredge] (v5)--(v8);
      \draw[nontouredge] (v6)--(v8);
      \draw[nontouredge] (v5)--(v9);
      \draw[nontouredge] (v6)--(v9);
      \draw[nontouredge] (v5)--(v10);
      \draw[nontouredge] (v6)--(v10);

      \draw[nontouredge] (v7)--(v10);
      \draw[nontouredge] (v8)--(v10);

      \draw[nontouredge] (v9)--(v11);
      \draw[nontouredge] (v10)--(v12);

      \draw[nontouredge] (v11)--(v5);
      \draw[nontouredge] (v11)--(v6);

      \draw[nontouredge] (v7)--(v11);
      \draw[nontouredge] (v7)--(v12);
      \draw[nontouredge] (v8)--(v11);
      \draw[nontouredge] (v8)--(v12);
    \end{scope}
  
  \end{tikzpicture}
  \caption{Counterexample to correctness of the canonical reduction
    from \TSP{}. The left shows an instance of \TSP{} along with an
    optimal tour (solid edges) of cost~$3 + \ell$. On the right, the
    resulting \CARP{} instance is shown (edge costs are as in their
    \TSP{} counterparts, solid thick edges have positive demand) along
    with an optimal tour of cost six (solid and thick edges). Scaling
    up~$\ell$ shows that the \CARP{} tour may be arbitrarily less
    costly than the \TSP{} tour.}
  \label{fig:counterex}
\end{figure}
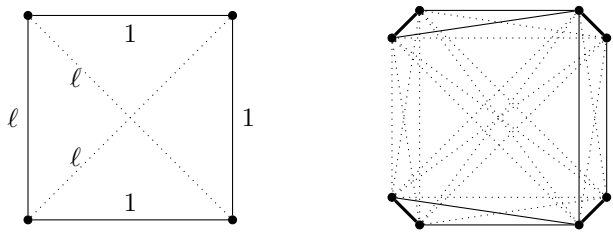

The reduction from \TSP{} to \CARP{} used by \citet{GolW81} is correct
when reducing from \TSP{} with triangle inequality.  In this case,
however, \TSP{} is factor-3/2 approximable using the algorithm by
\citet{Chr73} and, thus, the presented reduction does not imply
inapproximability for \CARP{}.

\section{\CF{} approximations for CARP\\without triangle
  inequality}
\noindent In the following, we show how to obtain \cf{} approximations for
\CARP{} without triangle inequality. Therein, we assume that the input graph is a complete graph,
since missing edges can be simulated by edges of cost~$\infty$.
First we adjust the edge costs so that the triangle inequality is
satisfied and then apply any \cf{} approximation algorithm
for \CARP{} with triangle inequality \citep{Jan93,Woe08c}.

\looseness=-1 In order to transform the edge costs so that the
triangle inequality holds, we set the cost of each edge~$\{u,v\}$ to
the cost of a shortest path between~$u$ and~$v$.  For zero-demand
edges, this transformation is correct: since edges can be traversed
more than once, instead of using a zero-demand edge, an optimal
solution can always use a shortest path between its endpoints.  For
edges with positive demand, the key idea is to visit these edges only
once when serving them.  Whenever the edge is traversed without
serving it, an optimal solution can again use the shortest path
between the endpoints.  We now formalize this idea.

We start by formally describing the transformation of an edge cost
function~$c$ into a cost function~$\chalftri{}$ that satisfies the
triangle inequality for zero-demand edges and then into a cost
function~$\cfulltri{}$ fully satisfying the triangle inequality.

\begin{definition}\label{def:distfuncs}
  Let $(G,c,d,W)$~be a \CARP{} instance.  We define the following
  modified edge cost functions.
  \begin{align*}
    \chalftri{}&\colon E(G)\to \mathbb N,\{u,v\}\mapsto
    \begin{cases}
      c(\{u,v\})&\text{ if }d(e)>0\\
      \dist_c(u,v)&\text{ otherwise},
    \end{cases}\\
    \cfulltri{}&\colon E(G)\to\mathbb N,\{u,v\}\mapsto\dist_c(u,v). 
  \end{align*}
  Herein, $\dist_c$~is the cost of a shortest path between~$u$ and~$v$
  with respect to the cost function~$c$. Finally, we use
  \begin{description}
  \item[$R:=\{e\in E(G)\mid d(e)>0\wedge \chalftri(e)\ne
    \cfulltri{}(e)\}$] to denote the set of positive-demand edges with costs exceeding the length of the shortest path between its endpoints and
  \item[$r:=\sum_{e\in R}(\chalftri{}(e)-\cfulltri{}(e))$] to denote
    the total cost decrease of the edges in~$R$
    from~$\chalftri{}$ to~$\cfulltri{}$.
  \end{description}
\end{definition}
\noindent It is easy to verify that~$\cfulltri{}$ satisfies the
triangle inequality.  Moreover, since any solution contains each edge
in~$R$ at least once, the following observation immediately follows.

\begin{observation}\label{obs:obsobs}
  Let $(G,c,d,W)$~be a \CARP{} instance.  Any feasible solution
  to~$(G,\chalftri{},d,W)$ of cost~$w$ has cost at most $w-r$
  in~$(G,\cfulltri{},d,W)$.
\end{observation}
\noindent Enforcing the triangle inequality on all edges with zero
demand does not change the cost of an optimal solution:

\begin{lemma}\label{lem:halftri}
  Let $(G,c,d,W)$~be a \CARP{} instance.
  \begin{enumerate}[i)]
  \item\label{halftri2} Any feasible solution for~$(G,c,d,W)$ is a
    feasible solution of at most the same cost
    for~$(G,\chalftri{},d,W)$ and
  \item\label{halftri1} any feasible solution
    for~$(G,\chalftri{},d,W)$ can be transformed into a feasible
    solution of the same cost for~$(G,c,d,W)$ in polynomial time.
  \end{enumerate}
\end{lemma}

\begin{proof}
  \eqref{halftri2} is trivial, since $\chalftri(e)\leq c(e)$ for all
  edges $e\in E(G)$.

  \eqref{halftri1} Let $(\mathcal C,s)$~be a feasible solution
  for~$(G,\chalftri{},d,W)$.  We obtain a modified set~$\mathcal C'$
  of cycles in polynomial time as follows. In each cycle~$C\in\mathcal
  C$, replace each edge~$\{u,v\}$ with~$d(\{u,v\})=0$ by a shortest
  path between~$u$ and~$v$ with respect to~$c$.  Then, $(\mathcal
  C',s)$ is a feasible solution for~$(G,c,d,W)$ since edges and
  vertices may be shared between cycles and may be used multiple
  times.  Moreover, by choice of~$\chalftri{}$, the cost of the
  cycles~$\mathcal C'$ with respect to~$c$ is the same as that
  of~$\mathcal C$.
\end{proof}
If we enforce the triangle inequality for all input edges, then we may
assume that an optimal solution uses every edge with positive demand
and modified cost at most once:

\begin{lemma}\label{lem:fulltri}
  Let $(G,c,d,W)$~be a \CARP{} instance.  Any feasible solution for
  $(G,\cfulltri{},d,W)$ can be transformed into a feasible
  solution~$(\mathcal C,s)$ with the same cost in polynomial time such
  that every edge in~$R$ is contained in exactly one cycle~$C$
  of~$\mathcal C$ and is contained in~$C$ exactly once.
\end{lemma}

\begin{proof}
  Observe that, for each edge~$e:=\{u,v\}\in R$, the condition
  $\chalftri{}(e)\ne \cfulltri{}(e)$ implies that there is a shortest path~$p_e$
  between~$u$ and~$v$ with respect to~$c$ that does not contain~$e$
  but has the same cost as~$e$ with respect to~$\cfulltri{}$.

  Thus, in any cycle~$C\in\mathcal C$ that does not serve~$e$, we
  simply replace any occurrence~$e$ by~$p_e$ without increasing the
  cost of~$C$ with respect to~$\cfulltri{}$.

  For the cycle~$C\in\mathcal C$ that serves~$e$, we replace all but
  one occurrence of~$e$ by~$p_e$, again without increasing the cost
  of~$C$ with respect to~$\cfulltri{}$.

  Clearly, these replacements work in polynomial time.
\end{proof}

\noindent We now prove \autoref{thm:mainthm}.
\addtocounter{theorem}{-1}

\mainthm{}

\begin{proof}
  Let $I:=(G,c,d,W)$~be a \CARP{} instance and let $\OPT$~denote the
  cost of an optimal solution.  The edge cost functions~$\chalftri{}$
  and~$\cfulltri{}$ can clearly be computed from~$c$ in polynomial
  time.  Thus, $I':=(G,\cfulltri{},d,W)$ is \pt{}
  computable and satisfies the triangle inequality.

  Let $(\mathcal C^*,s^*)$~be an optimal solution
  to~$(G,\cfulltri{},d,W)$ and let its cost be~$\OPTfulltri{}$.  By
  \autoref{lem:fulltri}, we may assume that $(\mathcal
  C^*,s^*)$~contains every edge of $R$ exactly once.  Hence,
  $(\mathcal C^*,s^*)$~is a solution of cost~$\OPTfulltri{}+r$
  for~$(G,\chalftri{},d,W)$ and, by \autoref{lem:halftri}, can be
  transformed into a solution~$(C',s'$) of cost~$\OPTfulltri{}+r$
  for~$(G,c,d,W)$.  Moreover, $(\mathcal C',s')$~is an optimal
  solution for~$(G,c,d,W)$ since, by \autoref{lem:halftri} and
  \autoref{obs:obsobs}, a cheaper solution of cost less
  than~$\OPTfulltri{}+r$ for~$(G,c,d,W)$ would imply a solution of
  cost less than~$\OPTfulltri{}$ for~$(G,\cfulltri{},d,W)$.  It
  follows that $\OPT\geq\OPTfulltri{}+r$.

  Now, assume that $(\mathcal C,s)$ is a solution for~$(G,\cfulltri{},d,W)$ of cost~$\alpha\cdot\OPTfulltri{}$.  We
  transform it into a solution of cost~$\alpha\cdot\OPT$ for~$(G,c,d,W)$ in
  polynomial time.  \autoref{lem:fulltri} allows us to assume
  that~$(\mathcal C,s)$ contains every edge of~$R$ exactly once, it
  follows that~$(\mathcal C,s)$~is a solution of cost~$\alpha\cdot
  \OPTfulltri{}+r$ for~$(G,\chalftri{},d,W)$ and, by
  \autoref{lem:halftri}, is \pt{} transformable into a
  solution of the same cost for~$(G,c,d,W)$.  Finally, since
  $\OPTfulltri{}+r\leq\OPT$, it follows that $(\mathcal C,s)$~has cost
  at most~$\alpha\cdot\OPT$.  
\end{proof}

\section{Conclusion}
\noindent We have shown that the triangle inequality is is not
necessary for finding good approximate solutions to \CARPlong{}, since
one can almost always replace an edge by a shortest path.

Our proof can be carried out analogously for variants of \CARP{} on
directed graphs: set the cost of any arc~$(u,v)$ to the cost of a
shortest directed path from~$u$ to~$v$.  However, it does not work for
graphs that have a mixture of directed and undirected edges, which
also appear in applications~\citep{Woe08}: it is not clear whether the
cost of an undirected edge~$\{u,v\}$ should be set to the length of a
shortest path from~$u$ to~$v$ or from~$v$ to~$u$.  It would be
interesting to show approximation results for this problem variant.

\section*{Acknowledgements}
\noindent We thank Sanne Wøhlk for valuable comments.
René van Bevern and Manuel Sorge acknowledge support of the Deutsche
Forschungsgemeinschaft (DFG), project DAPA (NI 369/12).
{
\small
\setlength\bibsep{0pt}
\bibliographystyle{abbrvnat}
\bibliography{tv_arp_ch02}

\newcommand{\noopsort}[1]{}
\begin{thebibliography}{11}
\providecommand{\natexlab}[1]{#1}
\providecommand{\url}[1]{\texttt{#1}}
\expandafter\ifx\csname urlstyle\endcsname\relax
  \providecommand{\doi}[1]{doi: #1}\else
  \providecommand{\doi}{doi: \begingroup \urlstyle{rm}\Url}\fi

\bibitem[{\noopsort{Bevern}van Bevern} et~al.(2014){\noopsort{Bevern}van
  Bevern}, Niedermeier, Sorge, and Weller]{BNSW13}
R.~{\noopsort{Bevern}van Bevern}, R.~Niedermeier, M.~Sorge, and M.~Weller.
\newblock Complexity of arc routing problems.
\newblock In {\'A}.~Corberán and G.~Laporte, editors, \emph{Arc Routing:
  Problems, Methods, and Applications}. SIAM, 2014.
\newblock In press.

\bibitem[Christofides(1973)]{Chr73}
N.~Christofides.
\newblock The optimum traversal of a graph.
\newblock \emph{Omega}, 1\penalty0 (6):\penalty0 719--732, 1973.

\bibitem[Edmonds(1975)]{Edm65}
J.~Edmonds.
\newblock The {C}hinese postman problem.
\newblock \emph{Operations Research}, pages B 73 -- B 77, 1975.
\newblock Supplement 1.

\bibitem[Edmonds and Johnson(1973)]{EdmJ73}
J.~Edmonds and E.~L. Johnson.
\newblock Matching, {E}uler tours and the {C}hinese postman.
\newblock \emph{Mathematical Programming}, 5:\penalty0 88--124, 1973.

\bibitem[Golden and Wong(1981)]{GolW81}
B.~L. Golden and R.~T. Wong.
\newblock Capacitated arc routing problems.
\newblock \emph{Networks}, 11\penalty0 (3):\penalty0 305--315, 1981.

\bibitem[Jansen(1993)]{Jan93}
K.~Jansen.
\newblock Bounds for the general capacitated routing problem.
\newblock \emph{Networks}, 23\penalty0 (3):\penalty0 165--173, 1993.

\bibitem[Lenstra and {Rinnooy Kan}(1976)]{LenK76}
J.~K. Lenstra and A.~H.~G. {Rinnooy Kan}.
\newblock On general routing problems.
\newblock \emph{Networks}, 6\penalty0 (3):\penalty0 273--280, 1976.

\bibitem[Sahni and Gonzalez(1976)]{SahG76}
S.~Sahni and T.~Gonzalez.
\newblock {P}-complete approximation problems.
\newblock \emph{Journal of the ACM}, 23\penalty0 (3):\penalty0 555--565, 1976.

\bibitem[Williamson and Shmoys(2011)]{WilS11}
D.~P. Williamson and D.~B. Shmoys.
\newblock \emph{The Design of Approximation Algorithms}.
\newblock Cambridge University Press, 2011.

\bibitem[Wøhlk(2008{\natexlab{a}})]{Woe08}
S.~Wøhlk.
\newblock A decade of {Capacitated Arc Routing}.
\newblock In B.~Golden, S.~Raghavan, and E.~Wasil, editors, \emph{The Vehicle
  Routing Problem: Latest Advances and New Challenges}, volume~43 of
  \emph{Operations Research / Computer Science Interfaces}, pages 29--48.
  2008{\natexlab{a}}.

\bibitem[Wøhlk(2008{\natexlab{b}})]{Woe08c}
S.~Wøhlk.
\newblock An approximation algorithm for the capacitated arc routing problem.
\newblock \emph{The Open Operational Research Journal}, 2:\penalty0 8--12,
  2008{\natexlab{b}}.

\end{thebibliography}
}

\end{document}